\documentclass{article}
\usepackage{arxiv}
\usepackage{url}

\usepackage[utf8]{inputenc} 
\usepackage[T1]{fontenc}    
\usepackage{hyperref}       
\usepackage{url}            
\usepackage{booktabs}       
\usepackage{amsfonts}       
\usepackage{amsmath}
\usepackage{amssymb}
\usepackage{amsthm}
\usepackage{color}
\usepackage[numbers]{natbib}
\usepackage{graphicx}
\usepackage[FIGTOPCAP]{subfigure}
\usepackage{braket}

\newtheorem{lemma}{Lemma}

\newcommand{\comment}[1]{}

\newcommand{\idty}[1]{\mathbb{1}}
\newcommand{\ovsqrt}[1]{\frac{1}{\sqrt{2}}}

\newcommand{\tr}[1]{\mathrm{Tr}\left( #1 \right)}

\usepackage{nicefrac}       
\usepackage{microtype}      
\usepackage{lipsum}
\newcommand{\Herm}{M}

\usepackage{soul}

\title{An initialization strategy for addressing barren plateaus in parametrized quantum circuits}

\author{
    Edward Grant \thanks{Corresponding author}\\
    Rahko Limited \&\\
    Department of Computer Science\\ 
    University College London\\
    \texttt{edward.grant@rahko.ai}
    \And
    Leonard Wossnig\\
    Rahko Limited \&\\
    Department of Computer Science\\
    University College London\\
    \texttt{leonard.wossnig@rahko.ai}
    \AND
    Mateusz Ostaszewski\\
    Institute of Theoretical and Applied Informatics\\
    Polish Academy of Sciences\\
    \texttt{mostaszewski@iitis.pl}
    \And
    Marcello Benedetti\\
    Cambridge Quantum Computing Limited \&\\
    Department of Computer Science\\
    University College London\\
    \texttt{m.benedetti@cs.ucl.ac.uk}
}

\begin{document}
\maketitle

\begin{abstract}

Parametrized quantum circuits initialized with random initial parameter values are characterized by barren plateaus where the gradient becomes exponentially small in the number of qubits. In this technical note we  theoretically motivate and empirically validate an initialization strategy which can resolve the barren plateau problem for practical applications.  The technique involves randomly selecting some of the initial parameter values, then choosing the remaining values so that the circuit is a sequence of shallow blocks that each evaluates to the identity. This initialization limits the effective depth of the circuits used to calculate the first parameter update so that they cannot be stuck in a barren plateau at the start of training. In turn, this makes some of the most compact ans\"atze usable in practice, which was not possible before even for rather basic problems. We show empirically that variational quantum eigensolvers and quantum neural networks initialized using this strategy can be trained using a gradient based method.
\end{abstract}


\section{Introduction}
Parametrized quantum circuits have recently been shown to suffer from gradients that vanish exponentially towards zero as a function of the number of qubits. This is known as the `barren plateau' problem and has been demonstrated analytically and numerically~\cite{mcclean2018barren}. The implication of this result is that for a wide class of circuits, random initialization will cause gradient-based optimization methods to fail. Resolving this issue is critical to the scalability of algorithms such as the variational quantum eigensolver (VQE)~\cite{peruzzo2014variational,kandala2017hardware} and quantum neural networks (QNNs)~\cite{schuld2018circuit,grant2018hierarchical,chen2018universal}.

In Ref.~\cite{verdon2019} the author shows that the barren plateau problem is not an issue of the specifically chosen parametrization, but rather extends the result to any direction in the tangent space of a point in the unitary group. In the Appendix we re-derive this result, and show that the gradient of the scalar $\bra{0}U(\alpha)^{\dagger}HU(\alpha)\ket{0}$ for any Hermitian operator $H$, vanishes with respect to any direction on the unitary group. Similarly the variance of the gradient decreases exponentially with the number of qubits.

Notably, this does not preclude the existence of a parametrization that would allow for efficient gradient-based optimization. Finding such a parametrization seems a non-trivial, but important, task. Indeed, as argued in Ref.~\cite{mcclean2018barren}, the barren plateau problem affects traditional ans\"atze such as the unitary coupled cluster, when initialized randomly, even for a small number of orbitals. The authors leave it as an open question whether the problem can be solved by employing alternative hardware-efficient ans\"atze.

In this technical note, we take an alternative route. Instead of proposing a new ansatz, we present a solution based on a specific way of initializing the parameters of the circuit. The strategy resolves the problem using a sequence of shallow unitary blocks that each evaluates to the identity. This limits the effective depth of the circuits used to calculate the gradient at the first iteration and allows us to efficiently train a variety of parametrized quantum circuits.

\section{A quick recap of the barren plateau problem}

Here we briefly recapitulate the original barren plateau problem and its generalization. Details of the derivation can be found in the Appendix. A parametrized quantum circuit can be described by a sequence of unitary operations
\begin{align}
\label{eq:ansatz}
U({\boldsymbol{\theta }}) = U_L (\theta_L)W_L \cdots U_1 (\theta_1)W_1 = \mathop {\prod}\limits_{l = L}^1 U_l(\theta _l)W_l ,
\end{align}
where $U_l(\theta_l) = \exp(-i \theta_l V_l)$, $\theta_l$ is a real-valued parameter, $V_l$ is an Hermitian operator, and $W_l$ is a fixed unitary. The objective function of a variational problem can be defined as the expectation
\begin{align}
E({\boldsymbol{\theta }}) = \langle 0|U({\boldsymbol{\theta }})^\dagger HU({\boldsymbol{\theta }})|0\rangle ,
\end{align}
where $H$ is an Hermitian operator representing the observable of interest. The partial derivatives take the form
\begin{align}
\label{eq:van_grad_1}
\partial_{\theta_k}E(\boldsymbol{\theta})  = i\left\langle {0\left| {U_ - ^\dagger \left[ {V_k,U_ + ^\dagger HU_ + } \right]U_ - } \right|0} \right\rangle,
\end{align}
where $U_{-} \equiv \prod^{1}_{l=k-1}U_{l}(\theta_{l})W_l$, and $U_{+} \equiv \prod^{k}_{l=L}U_{l}(\theta_{l})W_l$. If either $U_{-}$ or $U_{+}$ matches the Haar distribution~\cite{puchala2017symbolic} up to the second moment, e.g., $2$-designs~\cite{ambainis2007quantum}, the expected number of samples required to estimate Eq.~\eqref{eq:van_grad_1} is exponential in the system size.

In the Appendix we re-derive the result of Ref.~\cite{verdon2019} for the more general barren plateau problem. Concretely it is shown that the gradient in a direction $Z$ in the tangent space of the unitary group at $U(\alpha)$,
\begin{align}
\label{eq:van_grad_2}
\partial_\alpha E(U(\alpha))
= i \bra{0} Z^{\dagger}H U+U^{\dagger}H Z\ket{0},
\end{align}
where $U(0)=U$, and $\partial_\alpha U \rvert_{\alpha=0} = Z$, vanishes in expectation, i.e., 
\begin{align}
\mathbb{E}\left[ \partial_\alpha E(U(\alpha)) \right]
&= 0 .
\end{align}
Noting that the direction always takes the form $Z=-iUM$ for some fixed Hermitian matrix $M$, it is further shown that the variance
\begin{align}
\mathrm{Var}\left[ \partial_{\alpha} E(U(\alpha)) \right] &=2\frac{(\Herm^2)_{00}-(\Herm_{00})^2}{N^2-1}\left(\tr{H^2}-\frac{\tr{H}^2}{N} \right),
\end{align}
with $M_{00}:= \bra{0}M\ket{0}$ and $(M^2)_{00}:= \bra{0}M^2\ket{0}$, becomes exponentially small in the number of qubits.

This leads us to believe that the choice of parametrization for quantum circuits is a highly non-trivial task that can determine the success of the variational algorithm. In the next Section, we describe a method which resolves the barren plateau problem for practical purposes. In Section~\ref{s:experiments} we give numerical evidence on two different use cases.

\section{Initializing a circuit as a sequence of blocks of identity operators}
\label{sec:initialization_strategy}

Intuitively, the initialization strategy is as follows: we randomly select some of the initial parameter values and choose the remaining values in such a way that the result is a fixed unitary matrix, i.e., a deterministic outcome such as the identity. Additionally, we initially ensure that when taking the gradient with respect to any parameter, most of the circuit evaluates to the identity, which restricts its effective depth. This initialization strategy is optimized in order to obtain a non-zero gradient for most parameters in the first iteration. Obviously, this does not \textit{a priori} guarantee that the algorithm stays far from the barren plateau. However, numerical results indicate that this is indeed the case and that this initialization strategy allows the circuit to be trained efficiently. This gives an immediate advantage compared to previously known methods, which generally do not allow for training any parameter without an exponential cost incurred through the required accuracy.

Concretely, to ensure that $U_{-}$ and $U_{+}$ do not approach $2$-designs, we initialize the circuit via $M$ blocks where each block is of depth $L$. Depth $L$ is chosen to be sufficiently small so that the blocks are shallow and cannot approach 2-designs. In the following we will consider any fixed gate, i.e., $W_l$ in Eq.~\eqref{eq:ansatz}, as a parametrized one in order to simplify the presentation. For any $m=1,\dots,M$, the corresponding block has the form
\begin{align}
    U_m(\boldsymbol{\theta}_m)= \prod_{l=L}^{1} U_l(\theta^m_{l,1})\prod_{l=1}^{L} U_l(\theta^m_{l,2}).
\label{eqn:IDBlocks}
\end{align}

While the initial parameter values for the $U_l(\theta_{l,1})$ can be chosen at random, the parameter values for  $U_l(\theta_{l,2})$ are chosen such that $U_l(\theta_{l,2}) = U_l(\theta_{l,1})^{\dagger}$. Each block, and thus the whole circuit, evaluates to the identity, i.e.,
\begin{align}
    U(\boldsymbol{\theta}^{init}) = \prod_{m=M}^{1} U_m(\boldsymbol{\theta}_{m}) = \prod_{m=M}^{1} \left( \prod_{l=L}^{1} U_l(\theta^m_{l,1}) \prod_{l=1}^{L} U_l(\theta^m_{l,1})^{\dagger} \right) = \prod_{m=M}^{1} I_m = I .
\label{eq:IDCircuit}
\end{align}

It is important to choose each block $U_m(\boldsymbol{\theta}_{m})$ to be sufficiently deep to allow entanglement as training progresses. Yet each block should be sufficiently shallow so that $U_m(\boldsymbol{\theta}_{m})$, considered in isolation, does not approach a $2$-design to the extent that the gradients would become impractically small.

In Ref.~\cite{mcclean2018barren} it was shown that the sampled variance of the gradient of a two-local Pauli term decreases exponentially as a function of the number of qubits, which also immediately follows from the barren plateau result. Furthermore, the convergence towards this fixed lower value of the variance was numerically shown  to be a function of the circuit depth. This implies that for blocks $U_m(\boldsymbol{\theta}_{m})$ of constant depth, the whole circuit $U(\boldsymbol{\theta}^{init})$ is not in a barren plateau, allowing us to estimate the gradient efficiently for the initial learning iteration.

The intuition behind the \textit{identity block strategy} is the following: changing a single parameter in one block means that the other blocks still act like the identity. Therefore even if the whole circuit is deep enough to potentially be a $2$-design, changing any parameter will yield a shallow circuit. Notably this holds only for the first training parameter update.

We now analyze in more depth the behaviour of the initialization at the level of each block. An interesting property is that the gradients for gates located \textit{away} from the center of the block, e.g., in the beginning or at the end of a block, will have a larger magnitude. The reason is that the circuits required for the estimation are further from being $2$-designs since they are more shallow, as shown by the following calculation
\begin{align}
    \partial_{\theta^m_{k,1}} U_m(\boldsymbol{\theta}_m) &=
    \left( \prod_{l=L}^{k+1} U_l(\theta^m_{l,1}) \right) \left(\partial_{\theta^m_{k,1}} U_k(\theta^m_{k,1}) \right) \left(\prod_{l=k-1}^1 U_l(\theta^m_{l,1}) \right) \left(\prod_{l=1}^L U_l(\theta^m_{l,2})\right) \\
    &= \left( \prod_{l=L}^{k+1} U_l(\theta^m_{l,1}) \right) \left(\partial_{\theta^m_{k,1}} U_k(\theta^m_{k,1}) \right) U_k(\theta^m_{k,2}) \left(\prod_{l=k+1}^L U_l(\theta^m_{l,2})\right)
\end{align}
where we used that $U_l(\theta^m_{l,2})=U_l(\theta^m_{l,1})^\dagger$. For a small index $k$ we see that the circuit becomes shallow and hence the gradient is expected to be larger. A similar calculation can be done for the gradient with respect to the second set of parameters, i.e., $\partial_{\theta^m_{k,2}} U_m(\boldsymbol{\theta}_m)$. In this case, for index $k$ close to $L$ we see that the circuit becomes shallow.

To summarize, we expect that parameters at the boundaries of the blocks to have gradient larger than those at the center. Notice that the gradient can still be zero if $H$ commutes with the gate, which is a consequence of Eq.~\eqref{eq:van_grad_1} and Eq.~\eqref{eq:van_grad_2}. However, this is generally unlikely, and can be resolved by applying a small deterministic entangling circuit. More concretely, to avoid such cases we can add a shallow entangling layer $B$ to the circuit, i.e., $U_{M} \cdots U_1 B$, where $U_i$ are the blocks described above. This also resolves the barren plateau problem for the training of variational quantum eigensolvers, which we discuss in more detail in the next Section.

\section{Experimental results}
\label{s:experiments}

\subsection{Initializing a parametrized quantum circuit}

In this experiment we show the scaling of the variance of the gradient as a function of the number of qubits for both the random initialization and identity block strategy. For the random circuit, we use the same ansatz as in Ref.~\cite{mcclean2018barren}, Fig. $2$, and the same $ZZ$ observable. Their ansatz consists of layers of single qubit rotations $\exp(-i \theta_l V_l)$ about randomly selected axes $V_l \in \{X,Y,Z\}$, followed by nearest neighbor controlled-$Z$ gates. We used a total of $120$ layers of this kind. For the identity block initialization, we employed a single block as described by Eq.~\eqref{eq:IDCircuit} with $M=1$ and $L=60$. This setting also accounts for a total of $2LM=120$ layers. In both cases, initial values for the free parameters were drawn from the uniform distribution $\mathrm{unif}(0,2\pi)$.

In Fig.~\ref{fig:initial}~(a) we compare the variance of $(\partial_{\theta_{1,1,1}} E)$, i.e., the gradient with respect to the first element of the first part of the first block, as a function of the number of qubits $n$, when the circuit is applied to the input state $\sqrt{H}^{\otimes{n}}\ket{0}$. Each point in the Figure was computed from $200$ circuits. When using the random initialization, the variance decreased exponentially with the number of qubits, reproducing the plateau behaviour which was described in Ref.~\cite{mcclean2018barren}. In contrast to this, the variance of the circuit initialized as an identity block was invariant to the system size.  

In Fig.~\ref{fig:initial}~(b) we again compare the variance of $(\partial_{\theta_{1,1,1}} E)$ as a function of system size when circuits are applied to random MNIST images, downsampled and normalized such that they constitute valid quantum states. This type of encoding is known as amplitude encoding and represents a realistic scenario for computer vision tasks performed on a quantum computer. Each point in the Figure was computed from $200$ circuits. Similar to the previous experiment, the variance of the gradient vanished with the system size when using the random initialization. In contrast, for circuits using the identity block initialization, the variance did not vanish exponentially with the system size, showing that the plateau was avoided at initialization. 

\begin{figure}[ht]
\centering
\subfigure[]{\includegraphics[width=80mm]{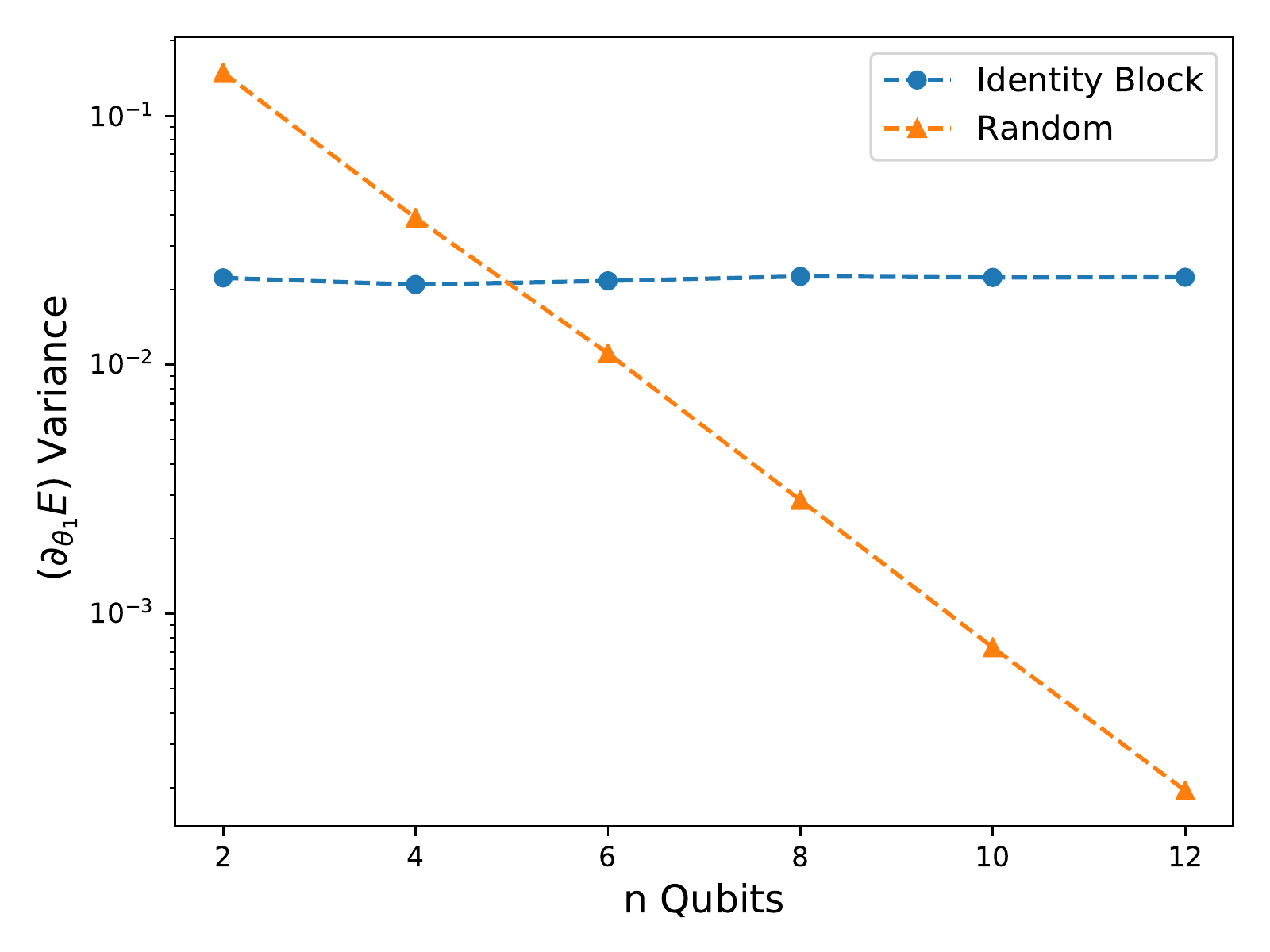}}
\subfigure[]{\includegraphics[width=80mm]{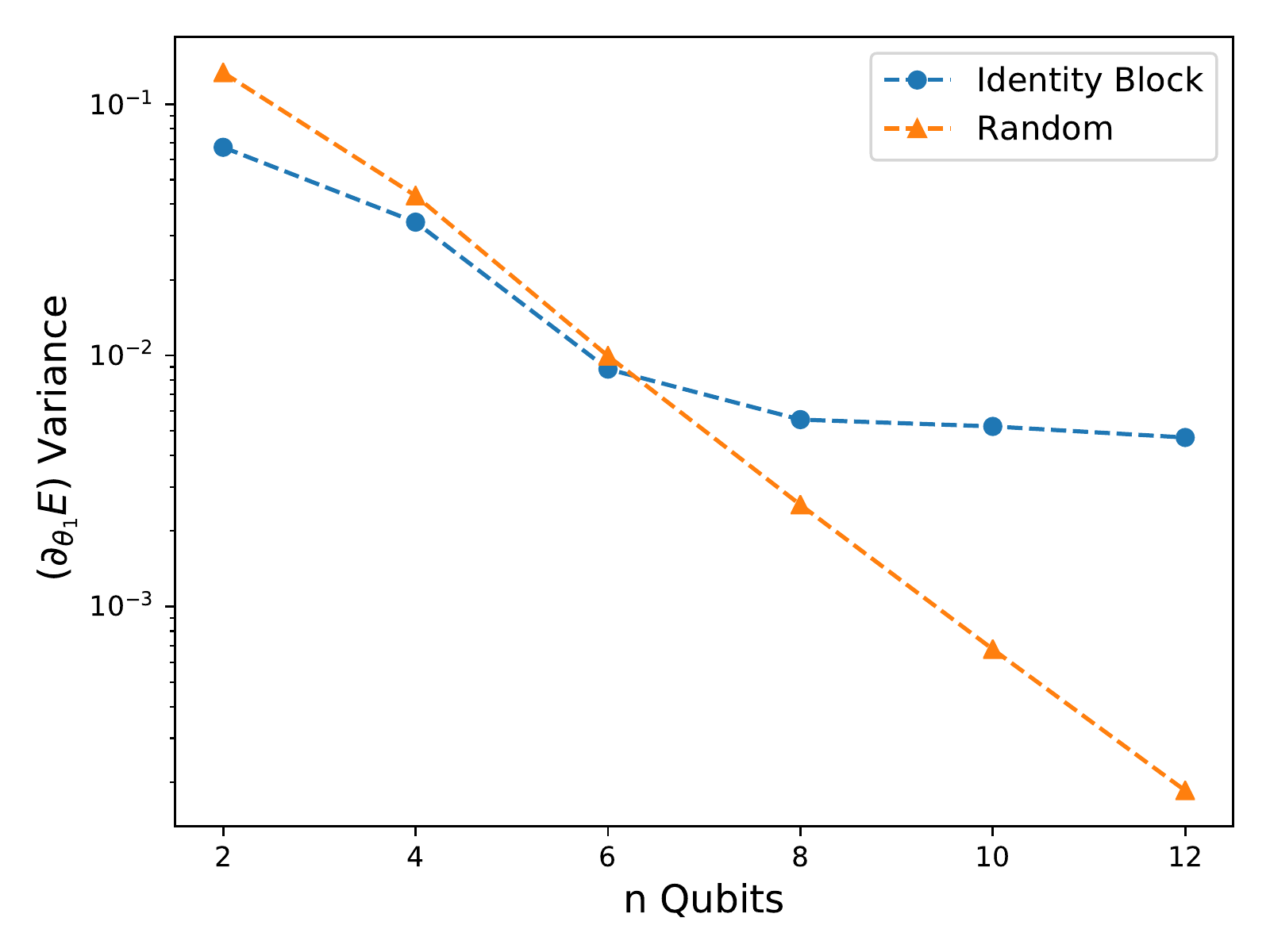}}
\caption{The variance of the gradient of the energy with respect to the first parameter as a function of number of qubits for (a) input state $\sqrt{H}^{\otimes{n}}\ket{0}$,  and (b) input state consisting of amplitude encoded MNIST images. Using the random initialization strategy resulted in the variance of the gradient becoming exponentially small as a function of the number of qubits. The identity block initialization prevented this for both input types (a) and (b).}
\label{fig:initial}
\end{figure}

\subsection{Training a quantum neural network classifier}

We have shown both analytically and empirically that a circuit with identity block initialization does not suffer from the barren plateau problem in the first training iteration. In this experiment we examine the variance of the gradient during training time to test whether a quantum neural network (QNN) classifier for MNIST images approaches the plateau. 

We used a $10$-qubit circuit with $M=2$ identity blocks, each having $L=33$ layers, for a total of $2LM=132$ layers. We selected $N=700$ MNIST images at random, resized them to $32 \times 32$, and finally reshaped them into vectors of dimension $2^{10}=1024$. We normalized each vector to unit length in order to be used as inputs to the circuit. Labels were set to $y_i=1$ for images of `even' digits, and $y_i=0$ for images of `odd' digits. For each MNIST example $\psi_i$, classification was performed by executing the circuit, measuring the observable $ZZ$, and finally ascribing a predicted probability to each class such that $P(even|\psi_i)=\tfrac{1}{2}(\bra{\psi_i}U(\theta)^{\dagger}ZZU(\theta)\ket{\psi_{i}}+ 1)$, and $P(odd|\psi_i)=1-P(even|\psi_i)$. The training was performed on $200$ different initial circuits, each constituting a different trial. We trained the circuit to minimize the binary cross-entropy loss
\begin{align}
\mathcal L = - \frac{1}{N} \sum_{i=1}^N y_i \log(P(even|\psi_i)) + (1-y_i) \log(1-P(even|\psi_i)) ,
\label{eq:cross_entropy}
\end{align}
Optimization was performed using the Adam optimizer with a learning rate of $0.001$~\cite{kingma2014adam} and a single randomly selected MNIST example used for each update. 

Figure~\ref{fig:training}~(a) shows the mean accuracy and standard deviation on a binarized MNIST dataset as a function of the training iterations for a circuit initialized using identity blocks compared with a strategy where all parameters are initially set to zero. While both strategies result in a circuit that initially evaluates to the identity, initializing all parameters to zero made the training much less efficient. Figure~\ref{fig:training}~(b) shows the variance of the partial derivative for parameters associated to the first three qubits, across different trials, and as a function of training iterations for circuits initialized with identity blocks. From the figures we observe that the model does not get stuck in a barren plateau (red dashed line), and that the variance decreases only as the model converges to a minimum of the objective function.

\begin{figure}[ht]
\centering     
\subfigure[]{\includegraphics[width=80mm]{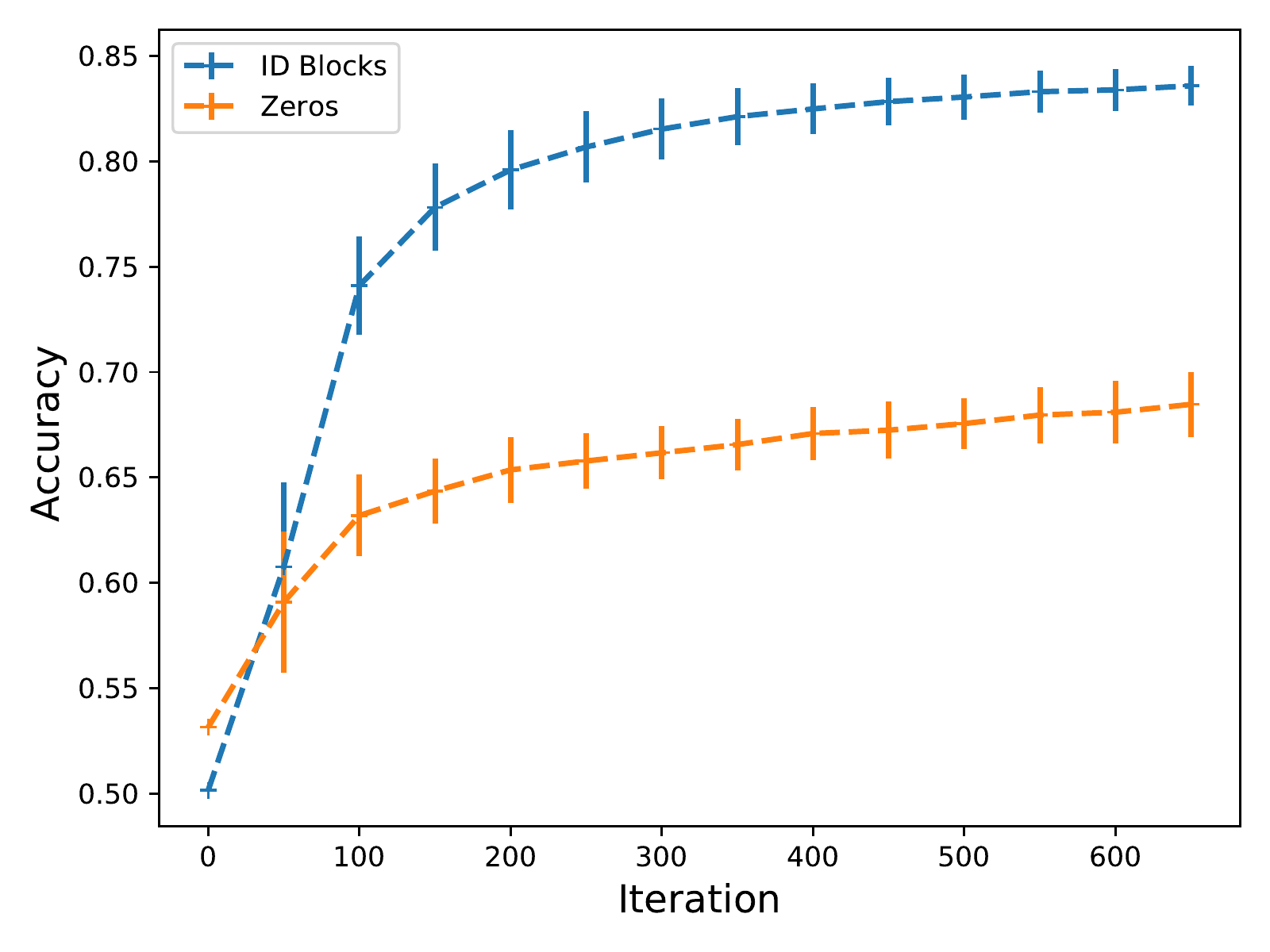}}
\subfigure[]{\includegraphics[width=80mm]{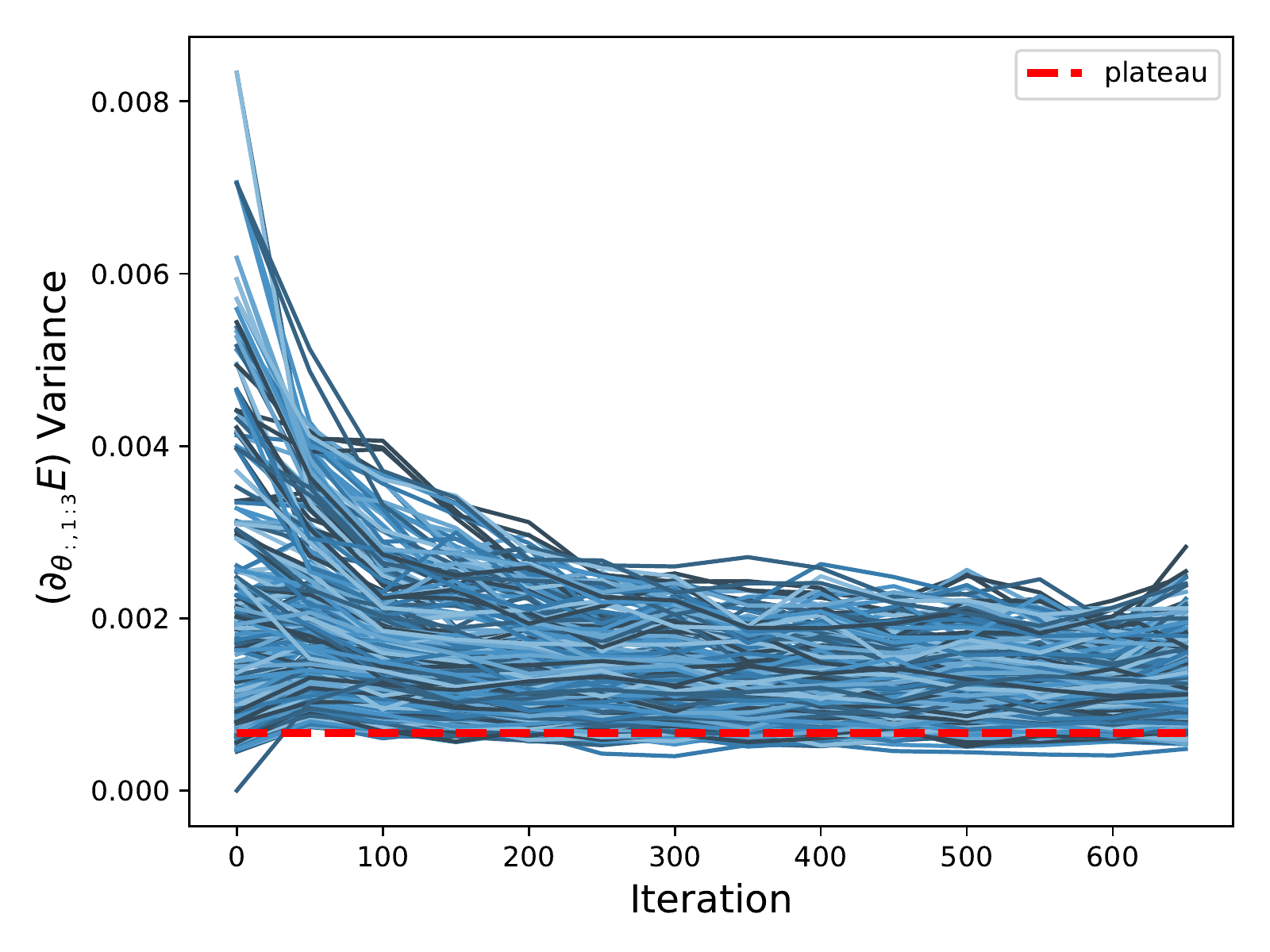}}
\caption{ (a) Training accuracy and standard deviation as a function of the number of training iterations for an MNIST classifying circuit initialized using two identity blocks compared with a circuit where all parameters are initially set to zero, and (b) variance across trials of the partial derivatives for parameters associated to the first three qubits in the circuit using the identity block initialization method.} The circuit initialized using identity blocks trained successfully in all trials and never encountered the barren plateau (red dashed line).
\label{fig:training}
\end{figure}

\subsection{Training a variational quantum eigensolver}
In this experiment we use the identity block strategy and train a variational quantum eigensolver (VQE) to find ground state energies. We chose the $7$-qubit Heisenberg model on a 1D lattice with periodic boundary conditions and in the presence of an external magnetic field. The corresponding Hamiltonian reads
\begin{align}
    H = J\sum\limits_{(i,j)\in \mathcal{E}}(X_iX_j+Y_iY_j+Z_iZ_j) + h\sum\limits_{i\in \mathcal{V}}Z_i,
\end{align}
where $\mathcal{G}=(\mathcal{V},\mathcal{E})$ is the undirected graph of the lattice with 7 nodes, $J$ expresses the strength of the spin-spin interactions, $h$ corresponds to the strength of the magnetic field in the $Z$-direction. In the experiments we set $J=h=1$. We chose this setting because for $J/h=1$ the ground state is highly entangled (see also Ref.~\cite{kandala2017hardware} for VQE simulations on the Heisenberg model).

Notice that the identity block initialization by itself cannot generate an entangled state at the very beginning. This could result in degraded performance for VQE. Hence, as the input state we chose $\ket{\psi}=B\ket{0}$, where $B$ consists of $7$ layers of random single-qubit rotations and controlled-$Z$ gates. These gates are never updated during training, and their purpose is to provide some level of entanglement even at the beginning. Training was performed using the Adam optimizer with a learning rate of $0.001$, until convergence. 

Figure~\ref{fig:VQE_init} shows the variance of the partial derivatives $(\partial_{\theta} E)$ across $200$ trials for (a) a circuit initialized using $M=2$ identity blocks and $L=33$ layers per block, and (b) a randomly initialized circuit with the same number of layers in total. In (b) we observe a barren plateau for all parameters, while in (a) we observe that most of the variances are well above the plateau. As expected, within each identity block the variance increases with distance from the center.  

Figure~\ref{fig:VQE_train}~(a) shows the mean and standard deviation of the expected energy across trials as a function of training iteration. A comparison is made between trials where circuits were initialized using identity blocks or where the initial parameters values were simply set to zero. Similar to the MNIST experiment, setting all parameters to zero resulted in poor training compared to the identity block strategy. The variance in the expectation when the zero parameter initialization strategy is used is explained by the random choice of the parameters in B for each trial. Panel~(b) shows the variance of the gradient across different trials as a function of training iterations for circuits initialized using identity blocks. The variance of the gradient approached zero during training as the model's energy (blue line) converged to the ground state energy (green dashed line). Thus, similarly to what observed in the MNIST classification experiment, the circuit did not get stuck in a plateau.

\begin{figure}[ht]
\centering
\subfigure[]{\includegraphics[width=80mm]{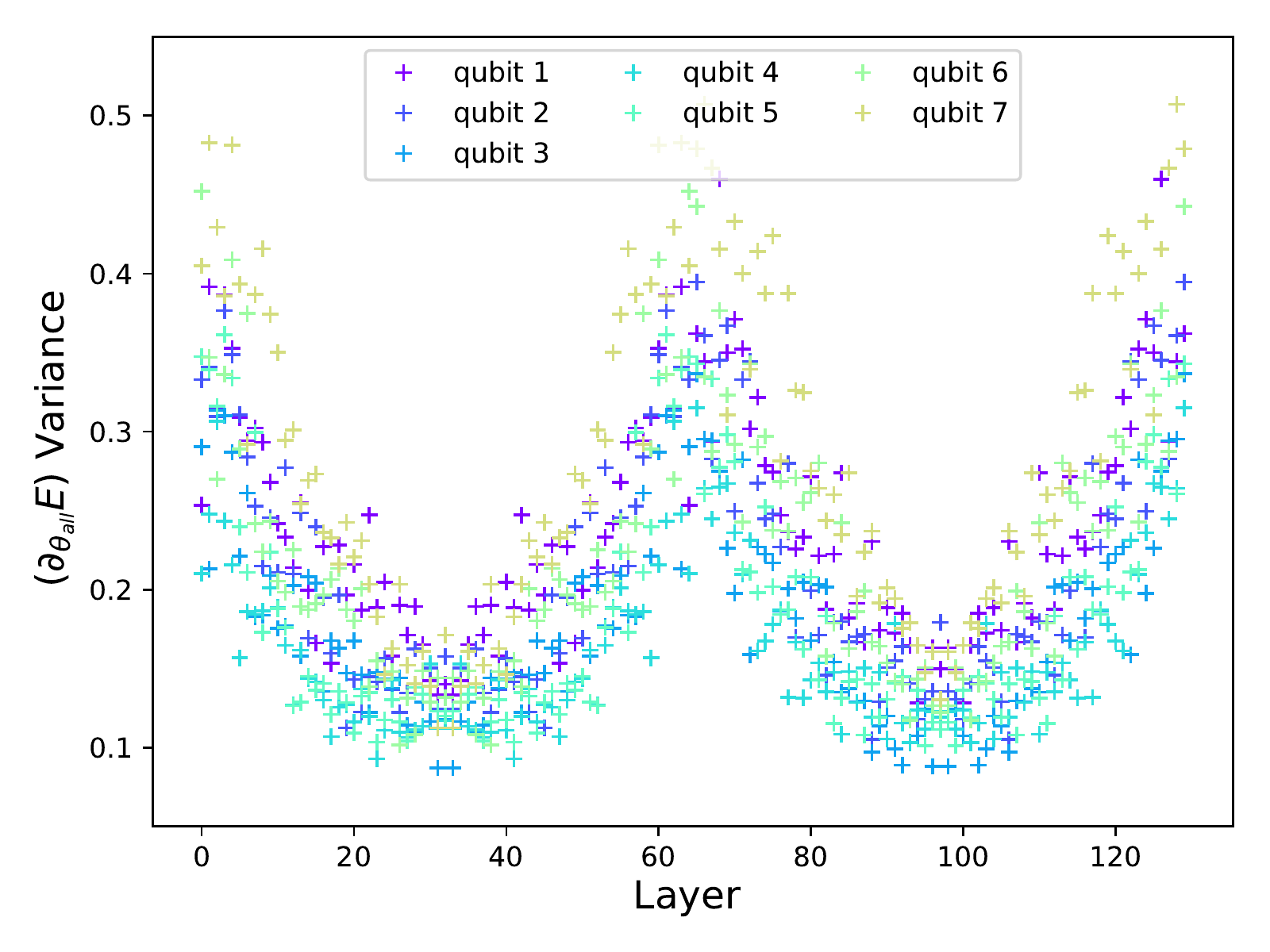}}
\subfigure[]{\includegraphics[width=80mm]{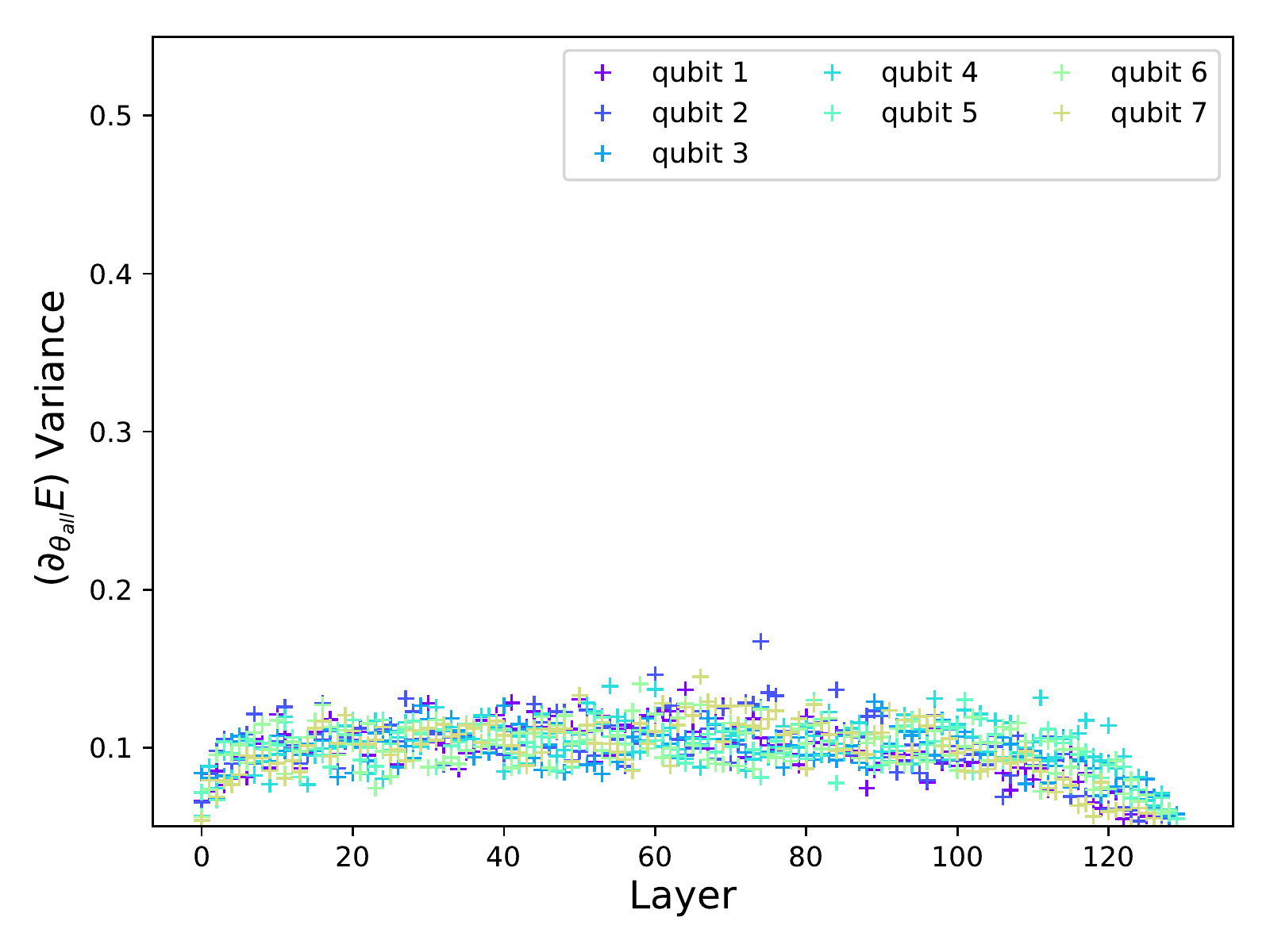}}
\caption{Variance of the gradient of a VQE circuit at the start of training across $200$ trials for (a) a circuit initialized using two identity blocks, and (b) a random initialization. In (a) the variance increases with the distance from the center of each identity block and the majority of variances exceed those in (b) where a barren plateau is expected.}
\label{fig:VQE_init}
\end{figure}

\begin{figure}[ht]
\centering     
\subfigure[]{\includegraphics[width=80mm]{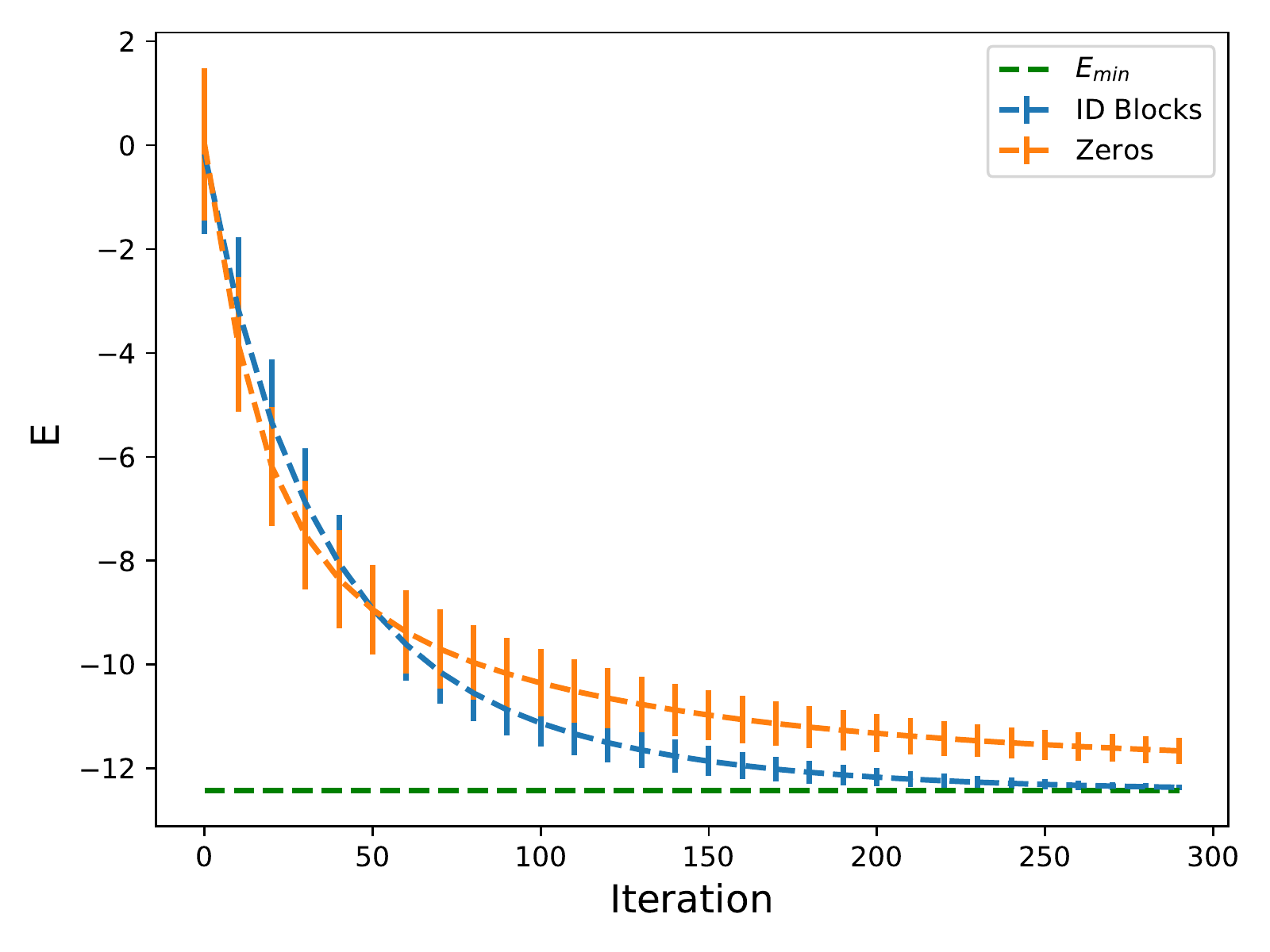}}
\subfigure[]{\includegraphics[width=80mm]{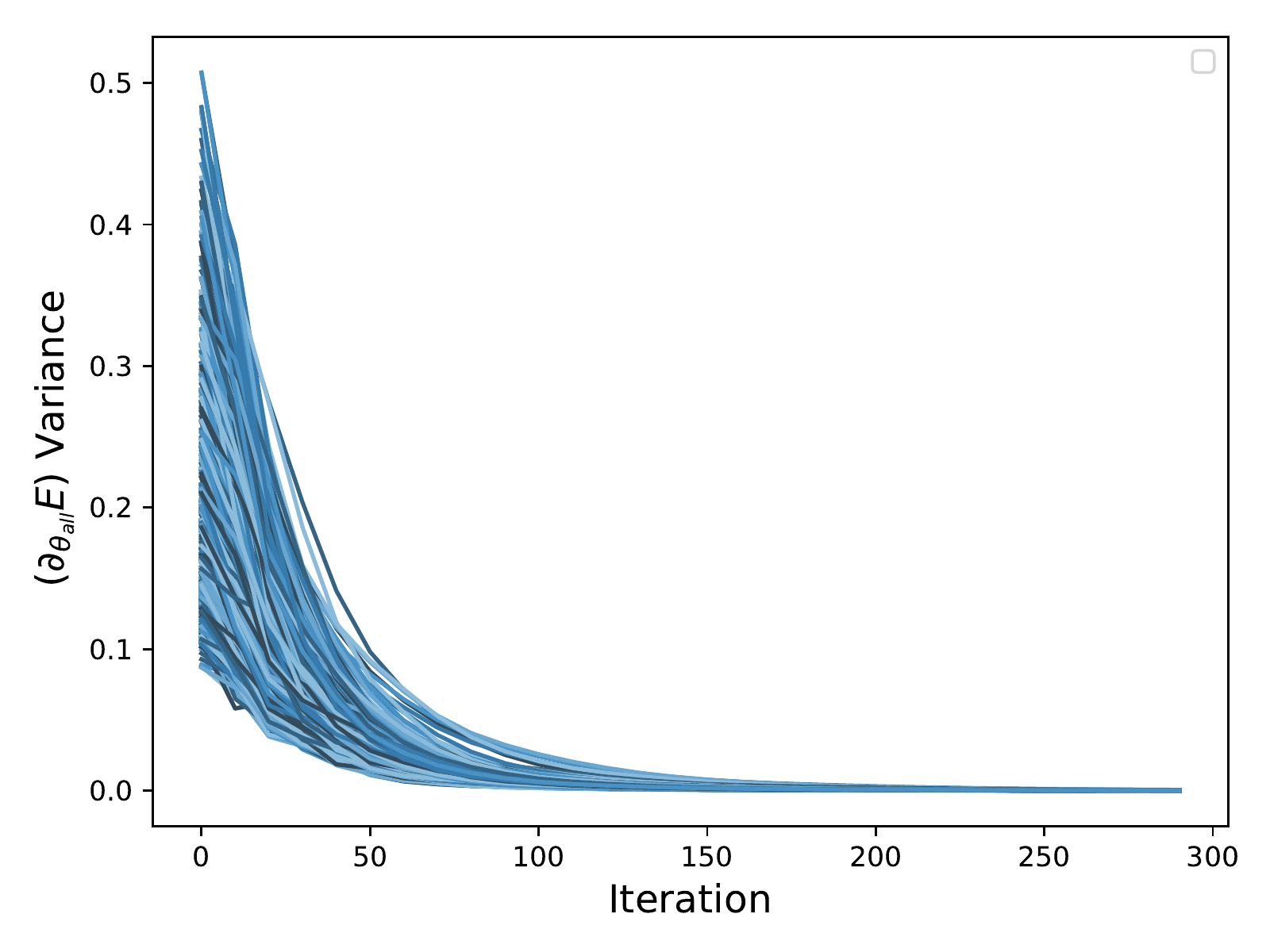}}
\caption{(a) Mean expected energy and standard deviation during training comparing a VQE initialized using identity blocks and a VQE initialized by setting all initial parameters to zero, and (b) variance of the gradient across trials as a function of training iteration for a VQE initialized using identity blocks. The circuit did not encounter a plateau during training since the variance of the gradient became small only as the model energy (blue line) converged to the ground state energy (green dashed line). In contrast to circuits initialized using identity blocks, circuits initialized by setting all parameters to zero failed to converge.}
\label{fig:VQE_train}
\end{figure}

\section{Conclusion}

In this technical note we motivated and demonstrated a practical initialization strategy that addresses the problem of barren plateaus in the energy landscape of parametrized quantum circuits. In the experiments we conducted, the \textit{identity block strategy} enabled us to perform well in two tasks: the variational quantum eigensolver (VQE), and the quantum neural network (QNN). 

More work is needed to assess the impact of input states and data encoding methods. In the case of VQEs, the strategy does not initially allow the circuit to generate an entangled state. We resolved this by adding a shallow entangling layer that is fixed throughout training. In the case of QNNs, the encoded input data can already be highly entangled, thereby reducing the depth of the circuit where the plateau problem occurs. From these examples we conclude that there is a problem-dependent trade-off to be analyzed. 

Finally, our approach is solely based on initialization of the parameter values. There are other potential strategies for avoiding barren plateaus such as layer-wise training, regularization, and imposing structural constraints on the ansatz. Understanding more about the relative merits of these and other approaches is a topic for future work. 

\section{Acknowledgements}
E.G. is supported by the UK Engineering and Physical Sciences Research Council (EPSRC) [EP/P510270/1]. L.W. is supported by the Royal Society. M.O. acknowledges support from Polish National Science Center scholarship 2018/28/T/ST6/00429. M.B. is supported by EPSRC and by Cambridge Quantum Computing Limited (CQC). We thank Raban Iten and Dominic Verdon for helpful technical discussions. We gratefully acknowledge the support of NVIDIA Corporation with the donation of the Titan Xp GPU used for this research.



\section{Appendix}

Here we provide a brief derivation of the vanishing gradient problem for the unitary group~\cite{verdon2019}.

\subsection{Vanishing gradient}
For Hermitian $H \in \mathbb{C}^{N \times N}$ and a normalized state $\ket{0} \in \mathbb{C}^N$, we consider the function $E(U)=\bra{0} U^{\dagger}H U\ket{0}$ for $U \in U(N)$, where $U(N)$ denotes the unitary group of dimension $N$. In the following, we calculate the derivative of $E(U)$ at a unitary $U$ in direction $Z$, where $Z$ lies in the tangent space at of the unitary group at the point $U$. To do so, we choose a path $U(\alpha)$ such that $U(0)=U$ and $\partial_\alpha U \rvert_{\alpha=0} = Z$. We have
\begin{align}
\partial_\alpha E(U(\alpha)) &=  \partial_\alpha \bra{0} U(\alpha)^\dagger H U(\alpha)\ket{0}\\
&=\bra{0} Z^{\dagger}H U+U^{\dagger}H Z\ket{0} .
\end{align}
We assume that $Z$ in the tangent space at $U$ has the form $Z=iU\Herm$ for some Hermitian operator $\Herm$, since it is easy to see that every $Z$ of this form is in the tangent space at $U$ and that every tangent vector at $U$ can be written in this form. 
Then, we have
\begin{align}
\partial_\alpha E(U(\alpha)) &= \bra{0} (-i\Herm U^{\dagger})H U+U^{\dagger}H (iU\Herm)\ket{0}\\
&=i\bra{0} U^{\dagger}H U\Herm-\Herm U^{\dagger}H U\ket{0}\\
&=i\bra{0} [U^{\dagger}H U,\Herm]\ket{0} \, .
\end{align}
Now, we would like to calculate the gradient over the whole unitary group. For this, we fix the Hermitian matrix $\Herm$ and find
\begin{align}
\mathbb{E}\left[ \partial_\alpha E(U(\alpha)) \right]
&=i \bra{0} [\mathbb{E}[U^{\dagger}H U],\Herm]\ket{0} \\
&=i \frac{\textnormal{Tr}(H)}{N} \bra{0} [I,\Herm]\ket{0} \\
&=0 \, ,
\end{align}
where we have used that for the Haar measure on the unitary group $\mu(U)$, $U \in U(N)$ it holds that ${\int d\mu(U) UOU^{\dagger} =  \frac{\tr{O}}{N} I}$, see \cite{puchala2017symbolic}.

Notably, if we initialize the matrix $U$ to be the identity for a fixed $H$, which could for example be achieved by just taking half the depth of the initial parametrized circuit $U_{1/2}$ and then appending the adjoint $U_{1/2}^{\dagger}$. The full initial circuit becomes 
\begin{align}
\label{eq:identity_init}
    U_{1/2}(\alpha) U_{1/2}^{\dagger}(\alpha)=I \, ,
\end{align}
then this is always the identity, i.e., constant. 
Plugging the identity into the expectation of the gradient, we then obtain 
\begin{align}
\label{eq:identity_init_2}
    \mathbb{E}\left[ \partial_\alpha E(U(\alpha)) \right]
    &=i \bra{0} [H,\Herm]\ket{0} \, ,
\end{align}
which is only zero whenever the Hamiltonian commutes with the observable, which is generally not the case. Note that this insight also holds for any other identity initialization such as the block initialization introduced in the body of the paper.

Note that trainable gates often take the form $\exp(-i \alpha_j V_j)$. If the $V_j$'s are chosen at random from tensor products of Pauli matrices $\{I,Z,X,Y\}^{\otimes n}$, then with high probability at least one of the derivatives is non-zero unless $H$ is the identity, see Eq.~\eqref{eq:identity_init_2}.
In sight of the initialization strategy, it is worth noting that initializing the circuit as $UU^\dag$ hence does not guarantee by itself that at least one derivative is non zero. 

\subsection{Vanishing variance}

We start with a simple identity.
\begin{lemma}
It holds that
\begin{align}
\left(U^{\dagger}AUBU^{\dagger}CU \right)_{ij} = \sum\limits_{k,l,m,n,p,q} A_{nm}B_{pq}C_{lk} U_{mp} U_{kj} U_{ni}^*  U_{lq}^* \, .
\end{align}
\end{lemma}
\begin{proof}
The proof follows from entry-wise evaluation.
\end{proof}
We further need the following identity for the second moments in the proof.
\begin{lemma}[\cite{puchala2017symbolic}]
For $\int d\mu(U)$ being the integral over the unitary group with respect to the random Haar measure, it holds that
\begin{align} \int d \mu(U) U_{i_1j_1}U_{i_2j_2}U^*_{i_1'j_1'}U^*_{i_2'j_2'} &= \nonumber 
\frac{\delta_{i_1i'_1} \delta_{i_2i'_2} \delta_{j_1j'_1} \delta_{j_2j'_2} + \delta_{i_1i'_2} \delta_{i_2i'_1} \delta_{j_1j'_2} \delta_{j_2j'_1}}{N^2-1} ~ -\nonumber \\ 
& \quad~ \frac{\delta_{i_1i'_1} \delta_{i_2i'_2} \delta_{j_1j'_2} \delta_{j_2j'_1} + \delta_{i_1i'_2} \delta_{i_2i'_1} \delta_{j_1j'_1} \delta_{j_2j'_2}}{N(N^2-1)} \, .
\end{align}
\end{lemma}

First observe that we can explicitly evaluate the variance and obtain
\begin{align}
    \mathrm{Var}\left[ \partial_{\alpha} E(U(\alpha)) \right] &=- \tr{\rho_0 \left(U^{\dagger}HU\Herm - \Herm U^{\dagger}HU\right)}^2 \\
    &= \tr{\rho_0 U^{\dagger}HU \Herm \rho_0 \Herm U^{\dagger}HU} \nonumber - \tr{\rho_0 U^{\dagger}HU \Herm \rho_0  U^{\dagger}HU \Herm} \nonumber + \\
    & \quad~ \tr{\rho_0 \Herm U^{\dagger}HU \rho_0  U^{\dagger}H U \Herm} - \tr{\rho_0 \Herm U^{\dagger}HU \rho_0 \Herm U^{\dagger}HU }.
\end{align}
Note that here we used the fact that the square of the trace is the trace of the square since $\rho$ is a rank one matrix, i.e., a projector.

We can proceed now by evaluating the expectation of each term individually. As an example we calculate the first term, since the remaining terms can be evaluated in a similar fashion.
We need to evaluate the following term
\begin{align}
    \int d \mu(U) \tr{\rho_0 U^{\dagger}HU \Herm\rho_0 \Herm U^{\dagger}HU} =  \tr{\rho_0 \underbrace{\int d \mu(U) U^{\dagger}HU \Herm\rho_0 \Herm U^{\dagger}HU}_{=: \xi }} .
\end{align}
Note that we can look at the integrand entry-wise before evaluating the trace. With $A = C :=H$, $B:= \Herm \rho_0 \Herm$, we have
\begin{align}
    \xi_{ij} &= \int d\mu([U]_{ij}) \sum_{k,l,m,n,p,q} A_{nm}B_{pq}C_{lk} U_{mp}U_{kj}U^*_{ni}U^*_{lq}  \\
    &=  \sum_{k,l,m,n,p,q} A_{nm}B_{pq}C_{lk} \int d\mu([U]_{ij}) U_{mp}U_{kj}U^*_{ni}U^*_{lq}  \\
    &=  \sum_{k,l,m,n,p,q} A_{nm}B_{pq}C_{lk} \left( \frac{\delta_{mn} \delta_{kl} \delta_{pi} \delta_{jq} + \delta_{ml} \delta_{kn} \delta_{pq} \delta_{ji}}{N^2-1} \right) ~ - \nonumber \\ 
    & \quad~ \sum_{k,l,m,n,p,q} A_{nm}B_{pq}C_{lk} \left(\frac{\delta_{mn} \delta_{kl} \delta_{pq} \delta_{ji} + \delta_{ml} \delta_{kn} \delta_{pi} \delta_{jq}}{N(N^2-1)}  \right) \\
    &= \frac{\tr{A}\tr{C}}{N^2-1}B_{ij} + \frac{\tr{B}\tr{AC}}{N^2-1}\delta_{ij} ~ - \nonumber \\
    & \quad~ \frac{\tr{A}\tr{B}\tr{C}}{N(N^2-1)}\delta_{ij} - \frac{\tr{AC}}{N(N^2-1)}B_{ij} .
\label{eq:evaluate_Int}
\end{align}

Note that plugging in $A,B$ and $C$ in Eq.~\eqref{eq:evaluate_Int}, then yields
\begin{align}
    \frac{\tr{H}^2}{N^2-1}(\Herm \rho_0 \Herm)_{ij} + \frac{\tr{\rho_0\Herm^2}\tr{H^2}}{N^2-1}\delta_{ij} - \frac{\tr{H}^2\tr{\rho_0 \Herm^2}}{N(N^2-1)}\delta_{ij} - \frac{\tr{H^2}}{N(N^2-1)}(\Herm \rho_0 \Herm)_{ij} .
\end{align}

Doing similar calculations for the other terms (using \eqref{eq:evaluate_Int} for different $A,B$ and $C$) and canceling and summarizing terms, yields the variance 
\begin{align}
    \mathrm{Var}\left[ \partial_{\alpha} E(U(\alpha)) \right] &=2\frac{(\Herm^2)_{00}-(\Herm_{00})^2}{N^2-1}\left(\tr{H^2}-\frac{\tr{H}^2}{N} \right) ,
\end{align}
where $H_{kl}$ denotes the $(k,l)$-entry of a matrix $H$ and $M_{00}:= \bra{0}M\ket{0}$, $(M^2)_{00}=\bra{0}M^2\ket{0}$. 
This indicates that the variance indeed also decreases exponentially with the number of qubits.

\end{document}